\documentclass[letterpaper,journal]{IEEEtran}

\usepackage[utf8]{inputenc} 
\usepackage[T1]{fontenc}
\usepackage{verbatim}

\usepackage[backend=bibtex, style = numeric, doi = false, url = false, isbn = false, maxbibnames = 6]{biblatex}
\bibliography{references.bib}
\renewbibmacro{in:}{}      
\newbibmacro{string+doi}[1]{\iffieldundef{doi}{#1}{\href{https://dx.doi.org/\thefield{doi}}{#1}}}
\DeclareFieldFormat{title}{\usebibmacro{string+doi}{\mkbibemph{#1}}}
\DeclareFieldFormat[article]{title}{\usebibmacro{string+doi}{\mkbibquote{#1}}}
\DeclareFieldFormat[incollection]{title}{\usebibmacro{string+doi}{\mkbibquote{#1}}}                   
\DeclareFieldFormat[inproceedings]{title}{\usebibmacro{string+doi}{\mkbibquote{#1}}}     

\usepackage[colorlinks=true,linkcolor=blue,urlcolor=blue,citecolor=blue,anchorcolor=green,pdfusetitle]{hyperref}

\usepackage[cmex10]{amsmath}  
\usepackage{amsfonts}
\usepackage{amssymb,amsthm,mathtools}

\usepackage{mleftright}       
\mleftright                   

\usepackage{graphicx}         
\usepackage{booktabs}         

\newtheorem{thm}{Theorem}

\newtheorem{lem}[thm]{Lemma}
\newtheorem{cor}[thm]{Corollary}
\theoremstyle{definition}

\usepackage{mathtools}
\usepackage{dsfont}
\usepackage{xcolor}

\newcommand{\cD}{\mathcal{D}}
\newcommand{\cE}{\mathcal{E}}
\newcommand{\cF}{\mathcal{F}}
\newcommand{\cH}{\mathcal{H}}

\newcommand{\cL}{\mathcal{L}}
\newcommand{\cl}{\mathrm{cl}}
\newcommand{\bC}{\mathbb{C}}
\newcommand{\mc}{\mathcal}
\newcommand{\mbb}{\mathbb}

\newcommand{\cB}{\mathcal{B}}
\newcommand{\sF}{\mathsf{F}}

\DeclareMathOperator{\tr}{tr}
\DeclareMathOperator{\id}{id}
 
\DeclareMathOperator{\supp}{supp}
\newcommand{\ket}[1]{|#1\rangle}
\newcommand{\bra}[1]{\langle #1|}
\newcommand{\op}[2]{|#1\rangle\langle #2|}

\newcommand{\one}{\mbb{I}}
\newcommand{\oneE}{\one}

\newcommand{\psucc}{p_{\mathrm{succ}}}

\newcommand{\Fcl}{\sF_{\mathrm{cl}}}

\definecolor{cool_green}{rgb}{0.0, 0.5, 0.0}

\usepackage[affil-it]{authblk}

\begin{document}

\title{On the Duality of Teleportation and Dense Coding}

\author[1,3]{Eric Chitambar\thanks{Email: \{\texttt{echitamb,leditzky}\}\texttt{@illinois.edu} }}
\author[2,3]{Felix Leditzky}

\affil[1]{Department of Electrical and Computer Engineering, University of Illinois Urbana-Champaign}
\affil[2]{Department of Mathematics, University of Illinois Urbana-Champaign}
\affil[3]{Illinois Quantum Information Science and Technology (IQUIST) Center, University of Illinois Urbana-Champaign}

\maketitle

\begin{abstract}
Quantum teleportation is a quantum communication primitive that allows a long-distance quantum channel to be built using pre-shared entanglement and one-way classical communication.
However, the quality of the established channel crucially depends on the quality of the pre-shared entanglement.
In this work, we revisit the problem of using noisy entanglement for the task of teleportation.  We first show how this problem can be rephrased as a state discrimination problem.  
In this picture, a quantitative duality between teleportation and dense coding emerges in which every Alice-to-Bob teleportation protocol can be repurposed as a Bob-to-Alice dense coding protocol, and the quality of each protocol can be measured by the success probability in the same state discrimination problem.  One of our main results provides a complete characterization of the states that offer no advantage in one-way teleportation protocols over classical states, thereby offering a new and intriguing perspective on the long-standing open problem of identifying such states. 
This also yields a new proof of the known fact that bound entangled states cannot exceed the classical teleportation threshold.  
Moreover, our established duality between teleportation and dense coding can be used to show that the exact same states are unable to provide a non-classical advantage for dense coding as well.  
We also discuss the duality from a communication capacity point of view, deriving upper and lower bounds on the accessible information of a dense coding protocol in terms of the fidelity of its associated teleportation protocol.
A corollary of this discussion is a simple proof of the previously established fact that bound entangled states do not provide any advantage in dense coding.
\end{abstract}

\section{Introduction}
Quantum teleportation \cite{bennett1993teleporting} is one of the most important protocols in quantum information theory.  It provides the backbone for a number of applications such as quantum communication \cite{Bennett-1996a}, long-distance entanglement distribution via repeaters \cite{Briegel-1998a}, and quantum computation \cite{Gottesman-1999a}.  From a fundamental perspective, teleportation beautifully demonstrates the interplay between different quantum resources: one communication primitive, a qubit channel, can be simulated using two other communication primitives, pre-shared entanglement and a two-bit classical channel.  Consequently, universal distributed quantum computation can be achieved using pre-shared entanglement along with local operations and classical communication (LOCC).  The rich study of entanglement and LOCC within quantum information science was largely due to this realization \cite{Horodecki-2009a}. 

The discovery of quantum teleportation was preceded by the dense coding protocol of Bennett and Wiesner~\cite{bennett1992communication}.
In this task, Alice and Bob use pre-shared entanglement and a noiseless qubit channel to simulate a two-bit classical channel. 
Thus teleportation and dense coding are dual to one another in terms of their resource consumption and communication objectives, an observation already made in  \cite{bennett1993teleporting}.
However, the duality extends even to the level of protocols.
As explained further in Sec.~\ref{sec:fidelity}, a general teleportation protocol (from Alice to Bob) consists of an entangled state shared between Alice and Bob, along with an encoding measurement performed by Alice and decoding state transformations applied by Bob.
As discussed in Sec.~\ref{sec:dense-coding}, these exact same elements (entangled state, measurement, state transformations) can also be used to define a dense coding protocol (from Bob to Alice), where now the state transformations serve as encoders of classical information while the measurement aims to decode this information. 

In the so-called ``tight'' setting (see Sec.~\ref{sec:dense-coding} for an explanation of this terminology) Werner~\cite{werner2001teleportation} showed that this duality leads to matching conditions of optimality: a teleportation protocol achieves perfect transmission of quantum information if and only if the associated dense coding protocol achieves perfect transmission of ``dense'' classical information. 
However, beyond this case of perfect transmission, a clear relation between figures of merit for these tasks has so far been missing.  
As one of our main results in this paper, we prove such a quantitative relationship in the general setting, recovering some of Werner's result along the way.

The foundational work of Horodecki et al.~\cite{horodecki1999singlet} established a direct relation between two quantities: the (average) fidelity of a teleportation protocol operating on bipartite states of equal local dimension, and the maximal ``singlet fraction'' of the entangled resource state.
The latter can be understood as a static version of the well-known entanglement fidelity of a quantum channel (see Sec.~\ref{sec:fidelity} for a definition).
The results of \cite{horodecki1999singlet} reduce the optimization of the average fidelity in teleportation to the potentially simpler optimization of the singlet fraction.
However, it obscures the duality between each teleportation protocol and its associated dense coding protocol, which we fully explore in this paper.


\subsection{Main results}

We first show in Lem.~\ref{lem:fidelity} that the entanglement fidelity of a general teleportation protocol is proportional to the success probability in an associated state discrimination problem defined in terms of the same data.
This generalizes the corresponding equivalence for port-based teleportation protocols \cite{ishizaka2008asymptotic,ishizaka2009quantum,beigi2011simplified} to completely general teleportation protocols.

Immediate corollaries of this result are a general bound on the fidelity in terms of the system dimensions (Cor.~\ref{cor:dim-bound}) and a short new proof of the known result \cite{horodecki1999singlet} that teleportation protocols operating on bound entangled states cannot exceed the classical fidelity achievable by any separable state (Cor.~\ref{cor:classical-fidelity}).

The latter statement is strengthened in Thm.~\ref{thm:beating-classical-limit}, in which we show that a bipartite state $\rho_{AB}$ can give rise to a non-classical teleportation fidelity if and only if there is a locally processed version of $\rho_{AB}$ violating the reduction criterion for separability.

We then use the operational connection between teleportation and state discrimination from Lem.~\ref{lem:fidelity} to prove a quantitative version of the duality between teleportation and dense coding.
This is achieved using two figures of merit for dense coding protocols: a classical analogue of the entanglement fidelity called classical correlation fidelity, and the accessible information of a quantum state ensemble.
We prove in Thm.~\ref{Prop:duality} an exact relationship between the classical correlation fidelity of a dense coding protocol and the entanglement fidelity of the associated teleportation protocol.
This result is used to show in Thm.~\ref{cor:beating-classical-limit-dense-coding} that a bipartite state gives rise to non-classical dense coding fidelity if and only if there is a locally processed version of the state violating the reduction criterion, in complete analogy to Thm.~\ref{thm:beating-classical-limit}.
Finally, in Thm.~\ref{thm:densecoding} we prove lower and upper bounds on the accessible information in terms of the fidelity of the corresponding teleportation protocol.
This result provides a new proof of some of the exact duality results of Werner~\cite{werner2001teleportation}, as well as a new proof of the result by Horodecki et al.~\cite{horodecki2001noisy} that bound entangled states do not provide an advantage in dense coding.

\subsection{Definitions and conventions}
Throughout the paper quantum systems are denoted by letters $A$, $B$, $C$, etc., and associated with finite-dimensional Hilbert spaces $\cH_A$, $\cH_B$, $\cH_C$, etc., whereas classical systems are denoted by letters $X$, $Y$, $Z$, etc.
The Hilbert space associated to a multipartite system $AB$ is defined as $\cH_{AB}\coloneqq \cH_A\otimes \cH_B$.
We use the notation $|A|\coloneqq \dim\cH_A$ for the dimension of a system $A$.
Two systems differing only by primes are assumed to be of the same dimension and hence isomorphic, e.g., $|C|=|C'|=|C''|$.
This also means that, given an operator $X_{C}$ acting on $C$, the operator $X_{C'}$ acting on $C'$ is also defined.
The space of linear operators acting on $\cH_A$ is denoted by $\cB(\cH_A)$.  
The identity operator on $\cH_A$ is denoted by $\oneE_A$, whereas the identity map on $\cB(\cH_A)$ is denoted by $\id_A$.
We denote by $X_{AB}^{T_{B}}$ the partial transpose of a bipartite operator.
A quantum state or density operator $\rho_A\in\cB(\cH_A)$ is a positive semidefinite operator of trace 1.
For a pure state $|\psi\rangle_A\in\cH_A$ we often write $\psi_A\equiv |\psi\rangle\langle\psi|_A$ for the associated density operator.
A quantum channel $\cE\colon A\to B$ is a linear, completely positive, trace-preserving map from $\cB(\cH_A)$ to $\cB(\cH_B)$.
We also use the alternative notation $\cE_{A\to B}$.
A positive operator-valued measure (POVM) on $A$ is a collection $\lbrace \Pi_A^i\rbrace_i$ of positive semidefinite operators satisfying $\sum_{i} \Pi_A^i = \oneE_A$.

\section{Teleportation and state discrimination}
\label{sec:fidelity}

\begin{figure*}
\centering
\includegraphics[width=\textwidth]{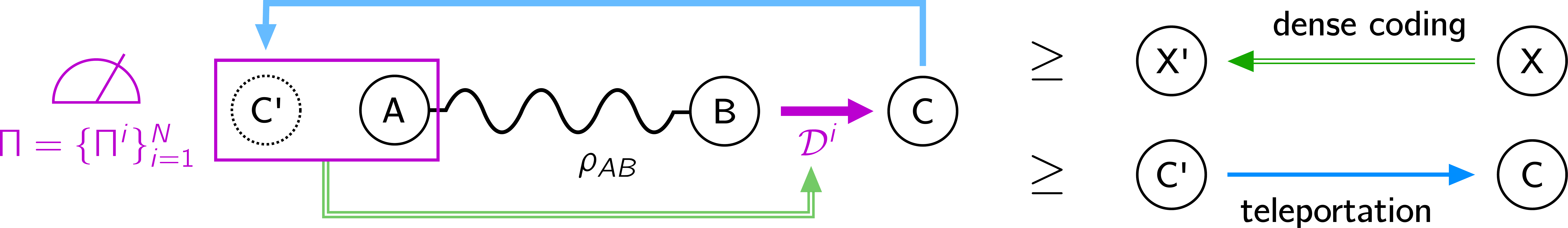} 
\caption{A graphical depiction of the operational duality of teleportation and dense coding, both operating on a given bipartite state $\rho_{AB}$ (wiggly line).
Read from left to right, the bottom half of the figure depicts a teleportation protocol, defined in terms of a measurement $\Pi = \lbrace \Pi^i\rbrace_{i=1}^N$ (purple box) and decoding operations $\lbrace \cD_{B\to C}^i\rbrace_{i=1}^N$ (purple arrow). 
Classical communication (lower green double arrow) goes from Alice (left) to Bob (right).
The protocol implements a quantum channel $C'\to C$ (lower blue arrow), indicated by the symbol $\geq$.
Read from right to left, the upper half of the figure depicts a dense coding protocol in terms of the same data, the decoding operations $\lbrace \cD_{B\to C}^i\rbrace$ and measurement $\Pi$.
The quantum channel (upper blue arrow) maps $C$ (Bob) to $C'$ (Alice).
It implements a dense classical channel $X\to X'$ (upper double green arrow), indicated by the symbol $\geq$. 
}
\label{fig:duality}
\end{figure*}

We start our discussion by showing that any teleportation protocol is equivalent to a quantum state discrimination problem: the fidelity in the former is proportional to the success probability in the latter.  This result generalizes the known relationship between fidelity and a state discrimination success probability for port-based teleportation, which was proved in \cite{ishizaka2008asymptotic,ishizaka2009quantum,beigi2011simplified}, and applied to generalized port-based teleportation protocols in \cite{kopszak2021multiport,strelchuk2021minimal}.

To prove our result, we first define general teleportation protocols, referring to Fig.~\ref{fig:duality} for a graphical depiction: Alice and Bob share a bipartite state $\rho_{AB}$, and in addition Alice controls a quantum system $C'$ in a state $\sigma_{C'}$ that she wants to teleport to Bob.
Note that the systems $A,B,C'$ are completely arbitrary and \emph{not} assumed to be of equal dimension.
Alice performs a measurement on $C'A$ given by a POVM $\Pi = \lbrace \Pi_{C'A}^i\rbrace_{i=1}^N$ for some $N\in\mathbb{N}$,
and communicates the outcome $i\in [N]\coloneqq \lbrace 1, \dots, N\rbrace$ to Bob over an $N$-message (noiseless) classical channel).
Finally, depending on the classical message $i$, Bob applies a decoding operation given by a quantum channel $\cD^i\colon B\to C$ to his system $B$.
This defines a $|C|$-dimensional teleportation protocol $(\rho_{AB},\lbrace\Pi^i\rbrace_{i=1}^N,\lbrace\cD^i\rbrace_{i=1}^N)$  
implementing the following quantum channel $\Lambda\colon C'\to C$, usually called \emph{teleportation channel}:
\begin{align}
	\Lambda(\sigma_{C'}) &= \sum_{i=1}^N (\tr_{C'A}\otimes \cD^i)\left[\left(\Pi_{C'A}^i\otimes \one_B\right) (\sigma_{C'}\otimes \rho_{AB}) \right]\notag \\
	&= \sum_{i=1}^N \tr_{C'A} \left[\left(\Pi_{C'A}^i\otimes \one_C\right) (\sigma_{C'}\otimes \omega^i_{AC}) \right],\label{eq:teleportation-channel}
\end{align}
where in the second line we defined the quantum states
\begin{align}
	\omega_{AC}^i = (\id_A\otimes \cD^i)(\rho_{AB}) \quad\text{for $i=1,\dots,N$.}\label{eq:omega-states}
\end{align}
The noise in a quantum channel is typically measured using the entanglement fidelity 
\begin{align}
	F\equiv F(\Lambda) \coloneqq \tr\left[ \Phi^+_{C''C} (\id_{C''}\otimes \Lambda)(\Phi^+_{C''C'}) \right],\label{eq:ent-fidelity}
\end{align}
where $|\Phi^+\rangle_{C''C} = |C|^{-1/2} \sum_{i=1}^{|C|} |i\rangle_{C''} |i\rangle_{C}$ for some orthonormal basis $\lbrace |i\rangle_C\rbrace_{i=1}^{|C|}$ is a maximally entangled state on $C''C$, and similarly for $C''C'$.
For simplicity, we will call $F=F(\Lambda)$ in \eqref{eq:ent-fidelity} simply the fidelity of a given teleportation protocol $(\rho_{AB},\lbrace\Pi^i\rbrace_{i=1}^N,\lbrace\cD^i\rbrace_{i=1}^N)$.

Our main result in this section relates this fidelity to the success probability of discriminating the $N$ states $\omega_{AC}^i$ defined in \eqref{eq:omega-states} with uniform prior probability using $\Pi = \lbrace \Pi^i\rbrace_{i=1}^N$.

\begin{lem}\label{lem:fidelity}
	Let $(\rho_{AB},\lbrace\Pi^i\rbrace_{i=1}^N,\lbrace\cD^i\rbrace_{i=1}^N)$ be a $|C|$-dimensional teleportation protocol as described above, giving rise to the teleportation channel $\Lambda$ defined in \eqref{eq:teleportation-channel}.
	Then the following relation holds:
	\begin{align}
 		F(\Lambda) = \frac{1}{|C|^2} \sum_{i=1}^N \tr\left(\Pi_{AC}^i \omega_{AC}^i\right) = \frac{N}{|C|^2} \psucc,\label{eq:fidelity}
	\end{align}
	where $\psucc = \dfrac{1}{N} \sum\limits_{i=1}^N \tr\left(\Pi_{AC}^i \omega_{AC}^i\right)$ is the success probability of the POVM $\Pi = \lbrace\Pi_{AC}^i\rbrace_{i=1}^N$ discriminating the states $\omega_{AC}^i$ drawn uniformly at random.
\end{lem}

\begin{proof}
	The proof is a direct generalization of the one given in \cite{beigi2011simplified} for port-based teleportation, and uses the transpose trick: Let $C''$ be a copy of the system $C$ with dimension $|C|$, and consider a maximally entangled state $|\Phi^+\rangle_{C''C'} = |C|^{-1/2}\sum_{i=1}^{|C|} |i\rangle_{C''}|i\rangle_{C'}$ on $C''C'$. 
	For any $X\in\cL(\cH_{C'})$ we have
	\begin{align}
		\begin{aligned}
			(\one_{C''}\otimes X_{C'})\Phi^+_{C''C'} &= (X_{C''}^T\otimes \one_{C'})\Phi^+_{C''C'}\\
			\Phi^+_{C''C'}(\one_{C''}\otimes X_{C'}) &= \Phi^+_{C''C'}(X_{C''}^T\otimes \one_{C'}).
		\end{aligned}
		\label{eq:transpose-trick}
	\end{align}
	We then use the expression \eqref{eq:teleportation-channel} of the teleportation channel $\Lambda$ to compute the entanglement fidelity in~\eqref{eq:ent-fidelity}.
	We suppress identity operators, identity maps, and certain tensor product symbols for better readability.
	\begin{align}
		F &= \tr\left[ \Phi^+_{C''C} \Lambda(\Phi^+_{C''C'}) \right]\\
		&= \sum_{i=1}^N \tr\left[ \Phi^+_{C''C} \tr_{C'A} \left(\Pi_{C'A}^i \left(\Phi^+_{C''C'}\otimes \omega^i_{AC}\right)\right) \right]\label{Eq:fidelity-eq2}\\
		&= \sum_{i=1}^N \tr\left[ \left(\Phi^+_{C''C} \otimes \Pi_{C'A}^i\right) \left(\Phi^+_{C''C'}\otimes \omega^i_{AC}\right) \right] \label{eq:partial-trace1}\\
		&= \sum_{i=1}^N \tr\left[ \Phi^+_{C''C} (\Pi_{C''A}^i)^{T_{C''}} \left(\Phi^+_{C''C'}\otimes \omega^i_{AC}\right) \right] \label{eq:transpose1}\\
		&= \sum_{i=1}^N \tr\left[ \Phi^+_{C''C} \Pi_{AC}^i \left(\Phi^+_{C''C'}\otimes \omega^i_{AC}\right) \right] \label{eq:transpose2}\\
		&= \frac{1}{|C|^2} \sum_{i=1}^N \tr\left(\Pi^i_{AC} \omega^i_{AC}\right). \label{eq:partial-trace2}
	\end{align}
	In \eqref{eq:partial-trace1} we used the identity $\tr(X_A Y_A)=\tr((X_A\otimes\one_B)Y_{AB})$, and in \eqref{eq:transpose1} and \eqref{eq:transpose2} we used the transpose trick \eqref{eq:transpose-trick} with respect to the $C'$ and $C''$ systems, respectively.
	Finally, \eqref{eq:partial-trace2} follows from taking partial traces first over $C'$ and then over $C''$, each time using the fact that $\tr_{C'}\Phi^+_{C''C'} = \frac{1}{|C|}\one_{C''}$, and similarly for $\Phi^+_{C''C}$.
	Using $p = \frac{1}{N} \sum_{i=1}^N \tr\left(\Pi_{AC}^i \omega_{AC}^i\right)$ establishes the second identity in \eqref{eq:fidelity}.
\end{proof}

The relation between fidelity and success probability in Lem.~\ref{lem:fidelity} is suggestive of the operational equivalence between teleportation and dense coding (see Fig.~\ref{fig:duality}): the states $\omega_{AC}^i$ result from Bob encoding classical information (the label $i$) in his system via the quantum operations $\cD^i\colon B\to C$.
Sending the $C$ system through the noiseless quantum channel to Alice, she aims to decode the message by performing a measurement $\Pi = \lbrace \Pi_{AC}^i\rbrace$.
We will prove a quantitative version of this operational equivalence in Sec.~\ref{sec:dense-coding} below.

A simple consequence of Lem.~\ref{lem:fidelity} is the following general upper bound on the fidelity of any teleportation protocol:
\begin{cor}
\label{cor:dim-bound}
	For any $|C|$-dimensional teleportation protocol with associated teleportation channel $\Lambda\colon C'\to C$,
 its entanglement fidelity $F\equiv F(\Lambda)$ is bounded from above as
	\begin{align}
		F \leq \frac{|A|}{|C|}.
	\end{align}
	This bound is saturated, $F=|A|/|C|$, if and only if $\omega_{AC}^i$ is pure and proportional to $\Pi^i_{AC}$ for all $i=1,\dots,N$.
\end{cor}

\begin{proof}
	The quantum states $\omega_{AC}^i = (\id_A\otimes \cD^i)(\rho_{AB})$ satisfy $\omega_{AC}^i \leq \one_{AC}$ for all $i=1,\dots, N$.
	Using this in the fidelity expression \eqref{eq:fidelity} gives
	\begin{align}
		F = \frac{1}{|C|^2} \sum_{i=1}^N \tr\left(\Pi_{AC}^i \omega_{AC}^i\right)
		&\leq \frac{1}{|C|^2} \sum_{i=1}^N \tr\left(\Pi_{AC}^i \right) \notag\\
		&= \frac{\tr\one_{AC}}{|C|^2} = \frac{|A|}{|C|},\label{eq:dim-bound}
	\end{align}
	where we used the completeness relation for the POVM $\lbrace\Pi_{AC}^i\rbrace_{i=1}^N$ in the second-to-last equality.
	
	Assume now that we have equality in eq.~\eqref{eq:dim-bound}, $\sum_{i=1}^N\tr (\Pi_{AC}^i \omega_{AC}^i )
	= \sum_{i=1}^N \tr (\Pi_{AC}^i )$, which can be rearranged to
	\begin{align}
		0 &= \sum_{i=1}^N \tr\left(\Pi^i_{AC} - \Pi^i_{AC}\omega_{AC}^i\right)\notag\\
		&= \sum_{i=1}^N \tr\left(\Pi^i_{AC} - (\Pi^i_{AC})^{1/2}\,\omega_{AC}^i\, (\Pi^i_{AC})^{1/2}\right).\label{eq:equality}
	\end{align}
	The operators $\Pi^i_{AC} - (\Pi^i_{AC})^{1/2}\,\omega_{AC}^i\, (\Pi^i_{AC})^{1/2}$ are positive semidefinite since $\omega^i_{AC}\leq \one_{AC}$, so that all traces in \eqref{eq:equality} are non-negative.
	Hence, the sum is zero iff $\tr\left(\Pi^i_{AC} - (\Pi^i_{AC})^{1/2}\,\omega_{AC}^i\, (\Pi^i_{AC})^{1/2}\right) = 0$ for all $i$.
	Since the trace of a positive semidefinite operator is zero iff the operator itself is zero, we have $\Pi^i_{AC} - (\Pi^i_{AC})^{1/2}\,\omega_{AC}^i\, (\Pi^i_{AC})^{1/2} = 0$, or in other words
	\begin{align}
		\omega_{AC}^i = P_{\Pi^i}.
	\end{align}
	Here, $P_{X}$ denotes the orthogonal projection onto the support $\supp(X)=(\ker X)^\perp$ of a Hermitian operator $X$.
	The $\omega^i_{AC}$ are normalized states, hence they can only be equal to a projection if their rank equals 1; in other words, $\omega^i_{AC}$ is pure for all $i$, and $\Pi^i_{AC} = \gamma_i \omega_{AC}^i$ for some $\gamma_i > 0$.
	The sufficiency of this condition for equality in \eqref{eq:dim-bound} is straightforward to check.
\end{proof}

The main use of Cor.~\ref{cor:dim-bound} lies in quantifying the deterioration of the teleportation fidelity of protocols attempting to teleport a $|C|$-dimensional system through an entangled state on $AB$ with $|C|\gg |A|$.
We also mention two situations in which the equality condition is satisfied:
First, let $|A|=1$ so that the resource state $\rho_{AB}$ is trivially separable.
In this case, choosing decoding operations $\cD^i_{B\to C}(X_B) = \tr(X_B)|i\rangle\langle i|_C$ where $\lbrace |i\rangle_C\rbrace_{i=1}^{|C|}$ is some orthonormal basis gives states $\omega_{AC}^i = |i\rangle\langle i|_C$.
With the POVM elements $\Pi^i_{AC} = \omega_{AC}^i = |i\rangle\langle i|_C$, we obtain a teleportation protocol achieving the classical teleportation fidelity $F = 1/|C|$; this is discussed in more detail in Sec.~\ref{sec:non-classical-teleportation}.
Another case of equality can be found when $|A|=|B|=|C|=d$ and $N=d^2$.
For a maximally entangled pure state $\rho_{AB}$ the $d^2$ Heisenberg-Weyl unitaries yield an orthonormal basis consisting of maximally entangled pure states $\omega_{AC}^i$ for $i=1,\dots, d^2$, which can be perfectly distinguished using the corresponding von Neumann measurement (i.e., we again have $\Pi^i_{AC}=\omega_{AC}^i$).
This is the $d$-dimensional version of the well-known perfect teleportation protocol \cite{bennett1993teleporting,werner2001teleportation}.
It is an interesting question to determine protocols based on partially entangled resource states $\rho_{AB}$ achieving a fidelity $1/|C| < F=|A|/|C| < 1$.

Traditionally, the figure of merit used to assess the quality of a $|C|$-dimensional teleportation channel $\Lambda$ has been the average transmission fidelity, 
$f(\Lambda):=\int d\psi\bra{\psi}\Lambda(\psi)\ket{\psi}$,
where $d\psi$ denotes the Haar integral over pure states $\ket{\psi}$ on $C$ \cite{Popescu-1994a}.
In the pioneering work of Horodecki et al.~\cite{horodecki1999singlet}, the average and entanglement fidelities were shown to be related as
\begin{align}
    f(\Lambda)=\frac{F(\Lambda)|C|+1}{|C|+1}.
\end{align}
Hence our results on the teleportation entanglement fidelity $F(\Lambda)$ can easily be translated into statements about the average teleportation fidelity $f(\Lambda)$.
We also note the recent independent work \cite{holdsworth2023quantifying} that derives an expression similar to eq.~\eqref{eq:fidelity} for the \emph{worst-case} fidelity of a general teleportation protocol.

\section{Non-Classical Teleportation and the Reduction Criterion}\label{sec:non-classical-teleportation}

In the distributed setting, $\rho_{AB}$ is typically considered a fixed shared resource, while Alice and Bob are free to choose the encoder $\{\Pi^i\}_i$ and decoders $\{\mc{D}^i\}_i$ in order to optimize the performance of their teleportation channel. 
Hence, it is natural to consider the quantity
\begin{align}
    \cF(\rho_{AB}; |C|)=\max_{\Lambda} F(\Lambda),
\end{align}
where the maximization is taken over all teleportation channels $\Lambda$ defined in terms of a $|C|$-dimensional teleportation protocol $(\rho_{AB},\lbrace \Pi^i\rbrace,\lbrace \cD^i\rbrace)$ via Eq.~\eqref{eq:teleportation-channel}. 
Horodecki et al.~\cite{horodecki1999singlet} showed that $\cF(\rho_{AB};|C|)$ is equivalent to the maximal singlet fraction of $\rho_{AB}$ obtained after processing $\rho_{AB}$ by one-way LOCC.
The fidelity expression in Lem.~\ref{lem:fidelity} reveals that $\cF(\rho_{AB}; |C|)$ involves a simultaneous maximization over encoder and decoders leading to a bilinear optimization problem, which is challenging to solve in general. 
In a future manuscript we will explore how upper bounds for $\cF(\rho_{AB}; |C|)$ can be efficiently computed, and are tight in some cases. 
We will also explore how symmetries of teleportation protocols can naturally be taken into account using the fidelity expression in \eqref{eq:fidelity}.

In the present work, we direct attention to the so-called classical lower bound on $\cF(\rho_{AB}; |C|)$.
Observe that by letting $\mc{D}^i$ be the CPTP that prepares the computational basis state $\op{i}{i}$ for $i=1,\cdots |C|$, i.e., $\mc{D}^i(X_B)=\tr[X_B]\op{i}{i}$ for all $X_B$, and by letting $\Pi^i_{AC'}=\mbb{I}_A\otimes\op{i}{i}_{C'}$ be the POVM effect that detects $\op{i}{i}_{C'}$, we obtain from Lem.~\ref{lem:fidelity} a fidelity
\begin{align}
	F = \frac{1}{|C|^2} \sum_{i=1}^{|C|} \tr\left(\Pi_{AC}^i \omega_{AC}^i\right)=\frac{1}{|C|^2} \sum_{i=1}^{|C|} \tr(\rho_A)=\frac{1}{|C|}.
\end{align}
Therefore, for all $\rho_{AB}$ we have
\begin{align}
\label{Eq:classical-bound}
\cF(\rho_{AB}; |C|)\geq \frac{1}{|C|} .
\end{align}
We refer to this as the \textit{classical bound} since it can be achieved using only classical resources; moreover, as we will prove below, it is tight whenever $\rho_{AB}$ is at most bound entangled, which includes all classical states.
Any state $\rho_{AB}$ for which $\cF(\rho_{AB}; |C|)> \frac{1}{|C|}$ can thus be deemed a legitimate quantum resource for teleportation, and we refer to any protocol exceeding the classical bound as an instance of \textit{non-classical teleportation}.\footnote{
	An alternative notion of non-classical teleportation has recently been proposed in the literature \cite{Cavalcanti-2017a}.  
	While this type of non-classicality is well-motivated from a quantum resource theory perspective, it refers to a task that is considerably different than simulating a noiseless quantum channel, which is the standard sense of teleportation and the object of study in this paper.
}
The fundamental question we consider here is what type of quantum states provide a resource for non-classical teleportation.  

To fully answer this question, let us first recall the reduction criterion for separability \cite{horodecki1999reduction}, which states that every separable (i.e., non-entangled) state $\rho_{AB}$ satisfies the operator inequality
\begin{align}
\rho_A\otimes\mbb{I}_B-\rho_{AB}\geq 0. \label{eq:reduction-criterion}
\end{align}
Even more, a violation of \eqref{eq:reduction-criterion} is a sufficient condition for $\rho_{AB}$ to be distillable \cite{horodecki1999reduction}.  It turns out that this inequality is also the key for identifying states supporting non-classical teleportation.  

As a warm-up, observe that Lem.~\ref{lem:fidelity} immediately implies that bound entangled states \cite{horodecki1998bound} cannot achieve strict inequality in the classical bound \eqref{Eq:classical-bound}.
To see this, let $\rho_{AB}$ be a bound entangled state, which is undistillable by definition.
Then the decoded states $\omega_{AC}^i = (\id_A\otimes \cD^i)(\rho_{AB})$ for $i\in[N]$ arising from a local quantum operation on the $B$ system are clearly also undistillable, so that they satisfy the reduction criterion, $\omega_{AC}^i \leq \omega_{A}^i \otimes \one_C = \rho_A \otimes \one_C.$
Crucially, the right-hand side of this operator inequality is independent of $i$, and so the sum in the fidelity expression \eqref{eq:fidelity} can be bounded as 
\begin{align}
	\sum\nolimits_i \tr\left(\Pi_{AC}^i \omega_{AC}^i\right) &\leq \sum\nolimits_i \tr\left(\Pi_{AC}^i  (\rho_A\otimes \one_C)\right)\notag\\ &=  \tr(\rho_A\otimes \one_C) = |C|.\label{eq:apply-reduction-criterion}
\end{align}
We have thus obtained an alternative proof for the well-known result \cite{horodecki1999singlet} that any teleportation protocol using a bound entangled state cannot exceed the classical bound.
\begin{cor}[\cite{horodecki1999singlet}]
	\label{cor:classical-fidelity}
	Any $|C|$-dim.~teleportation protocol using a bound entangled state cannot exceed a fidelity of $|C|^{-1}$.
\end{cor}

We now show that Cor.~\ref{cor:classical-fidelity} can be strengthened considerably, providing a complete characterization of states achieving non-classical teleportation.
\begin{thm}
\label{thm:beating-classical-limit}
	A bipartite state $\rho_{AB}$ can attain a (one-way) teleportation fidelity $\cF(\rho_{AB};|C|)$ exceeding the classical fidelity $|C|^{-1}$ 
	if and only if there exists a CPTP map $\mc{E}\colon B\to C$ such that $\omega_{AC}=\id_A\otimes\mc{E}(\rho_{AB})$ violates the reduction criterion. 
\end{thm}

\begin{proof}
	($\Rightarrow$) Suppose that $\omega_{AC}\leq \omega_A\otimes\oneE_C=\rho_A\otimes\oneE_C $ for all $\omega_{AC}=\id_A\otimes\mc{E}(\rho_{AB})$.  Hence for any teleportation protocol we have
	\begin{align}
		F=\frac{1}{|C|^2}\sum_i\tr[\Pi_i\omega_i] &\leq \frac{1}{|C|^2}\sum_i\tr[\Pi_i(\rho_A\otimes\oneE_C)] \notag\\ &=\frac{1}{|C|^2}\tr[\rho_A\otimes\oneE_C]=\frac{1}{|C|}.
	\end{align}
	
	($\Leftarrow$)  On the other hand, suppose there exists some CPTP map $\mc{E}\colon{B\to C}$ such that $\omega_{AC}$ violates the reduction criterion.  This requires the existence of a bipartite vector $\ket{\varphi}_{AC}$ such that $\bra{\varphi}\omega_{AC}\ket{\varphi}>\bra{\varphi}\rho_A\otimes\oneE_C\ket{\varphi}$.  We will use this inequality to construct a teleportation protocol that exceeds the classical threshold.  
	
	Let $\ket{\varphi}_{AC}=\sum_{i=0}^{r-1}\sqrt{\sigma_i}\ket{i}_A\ket{i}_C$ with $\sigma_i>0$ for $i=0, \dots, r-1$ be a Schmidt decomposition of $|\varphi\rangle$.  Introduce unitaries $U(n)$ that rotate the Schmidt basis $\{\ket{i}\}_i$ into a mutually unbiased basis,
	\begin{align}
		U(n)=\sum_{k=0}^{r-1}\omega^{kn}\op{k}{k},
	\end{align}
	with $\omega=e^{2\pi i/r}$ an $r$-th root of unity.  
	Observe the invariance $(U(n)\otimes U(n)^\dagger)\ket{\varphi}=\ket{\varphi}$ and define the states $\ket{\alpha_n}=(U(n)\otimes\oneE)\ket{\varphi}$ for $n=0,\dots,r-1$.  This implies that
	\begin{align}
		\bra{\alpha_n}\id_A\otimes\mc{U}_n(\omega_{AC})\ket{\alpha_n}>\bra{\alpha_n}\rho_A\otimes\oneE_C\ket{\alpha_n},
		\label{eq:phi-n-inequality}
	\end{align}
	where $\mc{U}_n(\rho)\coloneqq U(n)\rho U(n)^\dagger$.  
	A crucial property of the $\ket{\alpha_n}$ is that their ensemble average is classical,
	\begin{align}
		&\frac{1}{r}\sum_{n=0}^{r-1}\op{\alpha_n}{\alpha_n}\notag\\
		&=\frac{1}{r}\sum_{n=0}^{r-1}\sum_{k,k'=0}^{r-1}\omega^{n(k-k')}(\op{k}{k}\otimes\oneE)\op{\varphi}{\varphi}(\op{k'}{k'}\otimes\oneE)\notag\\
		&=\sum_{k=0}^{r-1}\sigma_k\op{k}{k}\otimes\op{k}{k},
	\end{align}
	where we used the orthogonality relation $\sum_{n=0}^{r-1} \omega^{n(k-k')} = r \delta_{k,k'}$ in the second equality.
	Hence, we can complete the set $\{\frac{1}{r}\op{\alpha_n}{\alpha_n}\}_{n=0}^{r-1}$ by the product operators
	\begin{align}
		\Pi_{r+n}=\sum_{\substack{k=1\\k\not=n}}^{|A|}\op{k}{k}\otimes\op{n}{n}+(1-\sigma_n)\op{n}{n}\otimes\op{n}{n}
	\end{align}
	for $n=0,\dots,r-1,$ so that $\sum_{n=0}^{r-1}(\tfrac{1}{r}\op{\alpha_n}{\alpha_n}+\Pi_{r+n})=\oneE_A\otimes\sum_{n=0}^{r-1}\op{n}{n}_C$.  Finally, let
	\begin{align}
		\Pi_{2r+n}=\mbb{I}_A\otimes\op{r+n}{r+n}
	\end{align}
	for $n=0,\cdots,|C|-r-1$.  
	The collection of positive semidefinite operators
	\begin{align}
		\left\{\tfrac{1}{r}\op{\alpha_n}{\alpha_n}\right\}_{n=0}^{r-1}\cup\{\Pi_{r+n}\}_{n=0}^{r-1}\cup \{\Pi_{2r+n}\}_{n=0}^{|C|-r-1}\label{Eq:POVM-protocol-full}
	\end{align}  
	constitutes a valid POVM.
	
	It remains to specify decoding operations for each operator corresponding to a measurement outcome:
	For each of the elements $\{\tfrac{1}{r}\op{\alpha_n}{\alpha_n}\}_{n=0}^{r-1}$ the corresponding decoder is the map $\mc{U}_n\circ\mc{E}$; for each of the elements $\{\Pi_{r+n}\}_{n=0}^{r-1}\cup \{\Pi_{2r+n}\}_{n=r}^{|C|-r-1}$ the corresponding decoder is the CPTP map that prepares the computational basis state $\op{n}{n}$, i.e., $\mc{D}_n(X)=\tr[X]\op{n}{n}$ for all $X$.  Since $\omega_{AC}=\mc{E}_{B\to C}(\rho_{AB})$, the fidelity of this protocol is given by
	\begin{align}
		&F=\frac{1}{|C|^2r}\sum_{n=0}^{r-1}\bra{\alpha_n}\id_A\otimes\mc{U}_n(\omega_{AC})\ket{\alpha_n}\notag\\
		&\quad {}+\frac{1}{|C|^2}\sum_{n=0}^{r-1}\tr[\Pi_{r+n}(\rho_A\otimes\op{n}{n})]+\frac{1}{|C|^2}\sum_{n=0}^{|C|-r-1}\tr[\rho_A]\notag\\
		&>\frac{1}{|C|^2r}\sum_{n=0}^{r-1}\bra{\alpha_n}\rho_A\otimes\oneE_C\ket{\alpha_n}\notag\\
		&\quad {} +\frac{1}{|C|^2}\sum_{n=0}^{r-1}\tr[\Pi_{r+n}(\rho_A\otimes\oneE_C)]+\frac{1}{|C|^2}\sum_{n=0}^{|C|-r-1}\tr[\rho_A]\notag\\
		&=\frac{1}{|C|^2}\tr[\rho_A\otimes\oneE_C]\notag\\
		&=\frac{1}{|C|},
	\end{align}
	where we used \eqref{eq:phi-n-inequality} in the inequality, proving the claim.
\end{proof}

%
Note that our definition of teleportation protocol includes any one-way LOCC pre-processing that might be performed on $\rho_{AB}$ before Alice and Bob actually decide to use their quantum resource for teleportation.
Given that  $\cF(\rho_{AB}; |C|)$ equivalently quantifies the largest singlet fraction obtainable by $\rho_{AB}$ after one-way LOCC processing \cite{horodecki1999singlet}, Thm.~\ref{thm:beating-classical-limit} thus says that $\rho_{AB}$ can achieve a singlet fraction exceeding $|C|^{-1}$ by one-way LOCC iff $\rho_{AB}$ can violate the reduction criterion after local processing on Bob's side. 

It is interesting to compare this result to the work of \cite{badziag2000local}, which shows that if a state's singlet fraction can exceed $|C|^{-1}$ after local processing on \textit{Alice's} side, then the original state $\rho_{AB}$ must already violate the reduction criterion.  In this case, the necessary condition stated in Thm.~\ref{thm:beating-classical-limit} for exceeding the classical bound $|C|^{-1}$  does not require an optimization over CPTP maps $\mc{E}\colon B\to C$.  However, in general such an optimization is needed.  Indeed, there are bipartite states that do not violate the reduction criterion, yet they can achieve a violation after processing on Bob's side.

Examples can be found in the family of Werner states \cite{werner1989epr}, defined to be $U\otimes U$-invariant bipartite states (for $U$ unitary).
More precisely, consider a Werner state on two $3$-dimensional systems,
\begin{align}
	\rho_\lambda=\frac{1}{24}[(3-\lambda)\oneE_3\otimes\oneE_3+(3\lambda-1)\mbb{F}_3],\label{eq:werner-state}
\end{align}
where $\oneE_3$ denotes the identity on $\mathbb{C}^3$ and $\mbb{F}_3$ denotes the swap operator acting on $\mathbb{C}^3\otimes \mathbb{C}^3$.
We have $\lambda = \tr(\rho_\lambda \mathbb{F})$, and $\rho_\lambda$ is entangled iff $\lambda<0$ \cite{werner1989epr}.
Furthermore, a direct calculation using the $U\otimes U$-symmetry shows that all Werner states on $\mathbb{C}^d\otimes \mathbb{C}^d$ satisfy the reduction criterion whenever $d>2$.

Let now $\lbrace |e_0\rangle,|e_1\rangle\rbrace$ and $\lbrace |f_0\rangle, |f_1\rangle,|f_2\rangle\rbrace$ be orthonormal bases for $\bC^2$ and $\bC^3$, respectively, and consider the quantum channel $\cE\colon \mc{B}(\bC^3)\to\mc{B}(\bC^2)$ with Kraus operators
\begin{align}
	P &= |e_0\rangle\langle f_0| + |e_1\rangle\langle f_1| & Q &= |e_0\rangle\langle f_2|.
\end{align}
Now the state $\sigma_\lambda = (\id\otimes \cE)(\rho_\lambda)\in\mc{B}(\bC^3\otimes \bC^2)$ violates the reduction criterion if and only if $\sigma_\lambda^{T_1} \ngeq 0$, where $T_1$ denotes the partial transpose on the first system \cite{horodecki1999reduction}.
Equivalently, we can check the eigenvalues of the operator 
\begin{multline}
	\id\otimes \cE \left(\rho_\lambda^{T_1}\right) = \frac{1}{24}\Big[(3-\lambda)\oneE_3\otimes(\oneE_2+\op{e_0}{e_0})\\
	{}+(3\lambda-1)(\phi_2^++\op{f_2}{f_2}\otimes\op{e_0}{e_0})\Big],
\end{multline}
where $|\phi_2^+\rangle = |f_0\rangle\otimes |e_0\rangle + |f_1\rangle \otimes |e_1\rangle$.
The smallest eigenvalue of this operator is equal to
\begin{align}
	7+3\lambda-\sqrt{13-30\lambda+37\lambda^2},
\end{align}
which is negative whenever $-1\leq\lambda < -3/7$.  
It thus follows that $\id\otimes\Lambda(\rho_\lambda)$ violates the reduction criterion in this range.

Theorem \ref{thm:beating-classical-limit} provides a criterion for testing whether a given bipartite state is useful for beating the classical bound in the one-way teleportation setting.  
The criterion can be explicitly expressed as a bilinear optimization problem, as we now show.  
For a given bipartite state $\rho_{AB}$, define its ``conditional state'' to be $\rho_{B|A}=\rho_A^{-1/2}\rho_{AB}\rho_{A}^{-1/2}$ \cite{Asorey2005,leifer2007conditional}, where we use the generalized inverse evaluated on the support of $\rho_A$.
The reduction criterion $\rho_{AB}\leq \rho_A\otimes \one_B$ can be rephrased in terms of the conditional state $\rho_{B|A}$ as $\rho_{B|A}\leq\mbb{I}_{\tilde{A}}\otimes\mbb{I}_B$, where $\tilde{A}$ is the support of $\rho_A$ (we will suppress the tilde from now on).
It follows that $\rho_{AB}$ violates the reduction criterion if and only if the largest eigenvalue of $\rho_{B|A}$ exceeds $1$.

Now consider any channel $\mc{E}_{B\to C}$ with Choi matrix $J_{BC}\coloneqq\id_B\otimes\mc{E}_{B'\to C}(\sum_{i,j=1}^{B}\op{ii}{jj}_{BB'})$ defined in terms of orthonormal bases $\lbrace |i\rangle_B\rbrace_{i=1}^{|B|}$ for $B$ and $\lbrace |j\rangle_{B'}\rbrace_{j=1}^{|B|}$ for $B'$.
By the Choi isomorphism, the action of $\mc{E}_{B\to C}$ on an input state $\chi_B$ is given by $\mc{E}_{B\to C}(\chi_B) = \tr_B\left(\chi_B^T J_{BC}\right)$, where $T$ denotes transposition in the chosen orthonormal basis for $B$ and we omitted identity operators for readability.
Hence, the action of $\mc{E}_{B\to C}$ on the $B$-part of $\rho_{B|A}$ can be written as 
\begin{align}
	\id_A\otimes \mc{E}_{B\to C}(\rho_{B|A}) = \tr_B \left(\rho_{B|A}^{T_B} J_{BC}\right).
\end{align}
Note that $\id_A\otimes \mc{E}_{B\to C}(\rho_{B|A})$ coincides with the conditional state of the output state $\id_A\otimes \mc{E}_{B\to C}(\rho_{AB})$.
The largest eigenvalue of $\id_A\otimes \mc{E}_{B\to C}(\rho_{B|A})$ can be expressed variationally as $\max_{\ket{\varphi}_{AC}}\bra{\varphi}\tr_B(\rho^{T_B}_{B|A}J_{BC})\ket{\varphi}_{AC}$, where the maximization is taken over all pure states on system $AC$.  
This maximum does not change when relaxing the optimization to all (not necessarily pure) density operators. 
Using the argument in the preceding paragraph, we thus arrive at the following optimization formulation of Theorem \ref{thm:beating-classical-limit}: a bipartite state $\rho_{AB}$ can exceed the $|C|$-dimensional classical teleportation bound iff $\lambda^*(\rho_{AB})>1$, where $\lambda^*(\rho_{AB})$ is defined as the solution to the following bilinear optimization problem:
\begin{align}
\text{maximize: }& \tr\left[\sigma_{AC}\tr_B\left(\rho^{T_B}_{B|A}J_{BC}\right)\right]\notag\\
\text{subject to: } &\tr_{C} J_{BC} =\mbb{I}_B;\notag\\
&\tr\sigma =1;\notag\\
&\sigma,J\geq 0.
\end{align}
Stated in this form, it becomes manageable to numerically explore the states which offer no non-classical advantage for teleportation, and we are currently pursuing such an investigation.


\section{Teleportation and dense coding}\label{sec:dense-coding}

We now make the operational duality between teleportation and dense coding concrete and quantitative.
To this end, we reinterpret a given $|C|$-dimensional teleportation protocol $(\rho_{AB},\lbrace\Pi^i\rbrace_{i=1}^N,\lbrace\cD^i\rbrace_{i=1}^N)$ as an $N$-message dense coding protocol from Bob to Alice, referring to Fig.~\ref{fig:duality} for a graphical depiction.
Bob encodes the message $i\in[N]$ into his system $B$ of the shared state $\rho_{AB}$ by applying $\cD^i\colon B\to C$, while Alice decodes the message by applying the POVM $\lbrace\Pi_{AC'}^i\rbrace_{i=1}^N$ to her systems $AC'$. 
Here it is assumed that a $|C|$-dimensional (noiseless) channel $\id_{C\to C'}$ connects Bob and Alice.
Such a protocol thus defines a classical channel $W\colon X\to X'$ characterized by the transition probabilities
\begin{align}
\label{Eq:transition-prob}
    p(j|i)=\tr(\Pi^j_{AC'}\omega^i_{AC'}),
\end{align}
with $\omega^i_{AC'}$ as defined in eq.~\eqref{eq:omega-states}.

We will discuss two different ways of assessing the quality of this classical channel, and hence the quality of a dense coding protocol.
The first one makes use of a classical correlation measure called ``classical correlation fidelity''.
The second one is by means of the classical capacity of the channel $W$, which coincides with the accessible information of the state ensemble defining the dense coding protocol.

The advantage of the classical correlation fidelity as a figure of merit is that it is a classical analogue of the entanglement fidelity in teleportation, allowing us to exhibit a concise quantitative operational duality between teleportation and dense coding.
On the other hand, the accessible information is an information-theoretic measure with a clear operational interpretation in terms of classical information transmission.
The two figures of merit can also be related to each other; see the discussion after Thm.~\ref{thm:densecoding}.

\subsection{Classical correlation fidelity}
For a quantum channel $\Lambda\colon C'\to C$, we can interpret $F(\Lambda)$ as a measure of how well $\Lambda$ preserves maximal \textit{coherent} correlations with a reference system $C''$.
For a classical channel $\tau_{\cl}\colon X\to X'$, we thus define a corresponding classical correlation measure $\sF(\tau_{\cl})$ called ``classical correlation fidelity'' that quantifies how well $\tau_{\cl}$ preserves maximal \textit{incoherent} correlations with a reference system $X''$:
\begin{align}
\sF(\tau_{\cl})\coloneqq \tr\left[\gamma^+_{X''X'}\id\otimes \tau_{\cl}(\gamma^+_{X''X})\right],
\end{align}
where $\gamma^+_{X''X}=\frac{1}{N}\sum_{i=1}^N\op{ii}{ii}_{X''X}$ is a perfectly correlated state on two $N$-level classical systems $X''$ and $X$.
We use a sans-serif font $\sF(\tau_{\cl})$ defined for a classical channel $\tau_{\cl}$ to distinguish this classical fidelity quantity from the quantum entanglement fidelity $F(\Lambda)$ defined for a quantum channel $\Lambda$.

Evaluating $\Fcl\equiv \sF(W)$ for the classical dense coding channel $W\colon X\to X'$ defined above with transition probabilities $p(j|i)$ given in eq.~\eqref{Eq:transition-prob}, we obtain
\begin{align}
    \Fcl=\frac{1}{N}\sum_{i=1}^Np(i|i) = \psucc,
\end{align}
where $\psucc$ is the success probability of discriminating the uniformly drawn states $\lbrace \omega_{AC}^i\rbrace$ using the POVM $\lbrace\Pi_{AC}^i\rbrace_{i=1}^N$, previously introduced in eq.~\eqref{eq:fidelity}.
This allows us to reformulate Lem.~\ref{lem:fidelity} as a quantitative link between between teleportation and dense coding protocols.
\begin{thm}
\label{Prop:duality}
Let $(\rho_{AB},\lbrace\Pi^i\rbrace_{i=1}^N,\lbrace\cD^i\rbrace_{i=1}^N)$ define either a $|C|$-dimensional teleportation protocol from Alice to Bob or an $N$-message dense coding protocol from Bob to Alice.
Then
\begin{align}
    F=\frac{N}{|C|^2}\Fcl,\label{eq:F-Fcl}
\end{align}
where $F=F(\Lambda)$ is the entanglement fidelity of the teleportation channel $\Lambda$ and $\Fcl = \sF(W)$ is the correlation fidelity of the classical dense coding channel $W$.
\end{thm}

Similar to teleportation, we can now define the measure
\begin{align}
    \cF_{\cl}(\rho_{AB};N, |C|)=\max_{W}\sF(W),\label{eq:opt-Fcl}
\end{align}
where the maximization is taken over all classical channels $W$ induced by an $N$-message dense coding protocol using a $|C|$-dimensional quantum channel connecting Bob to Alice.
Let us again identify a classical bound for this quantity.
Observe that, if $\rho_{AB}$ is separable, then so are the $\omega_{AC}^i$ states, and thus they satisfy the reduction criterion $\omega_{AC}^i \leq \omega_A^i \otimes \oneE_C = \rho_A\otimes \oneE_C$.
Using the same argument as in \eqref{eq:apply-reduction-criterion}, we then have 
\begin{align}
    \Fcl = \frac{1}{N}\sum_i p(i|i) \leq \frac{|C|}{N}.
\end{align}
In fact, this inequality holds for all bound entangled states.
On the other hand, the protocol described above eq.~\eqref{Eq:classical-bound} allows Bob to send $|C|$ messages to Alice over a $|C|$-dimensional quantum channel.
Therefore, we conclude that
\begin{align}
    \cF_{\cl}(\rho_{AB};N, |C|)\geq \min\left\{|C|/N,1\right\},
\end{align}
and this inequality is tight for all bound entangled states $\rho_{AB}$.
We identify this as the classical bound for dense coding, and any state for which the inequality is strict (for some values of $N$ and $|C|$) is a non-classical resource for dense coding.

Consider now any state satisfying the conditions of Thm.~\ref{thm:beating-classical-limit}.  
As shown in the proof of Thm.~\ref{thm:beating-classical-limit}, the teleportation protocol achieving $F>|C|^{-1}$ uses $N=|C|+r$ classical messages and a $|C|$-dimensional quantum channel.  
Hence, by Thm.~\ref{Prop:duality}, the correlation fidelity in the corresponding dense coding protocol satisfies
\begin{align}
    \Fcl =\frac{|C|^2}{N}F>\frac{|C|}{N}=\min\left\{|C|/N,1\right\}.
\end{align} 
We therefore obtain the following dual statement to Thm.~\ref{thm:beating-classical-limit}.
\begin{thm}
    \label{cor:beating-classical-limit-dense-coding}
    A bipartite state $\rho_{AB}$ can attain a dense coding fidelity $\cF_{\cl}(\rho_{AB};N, |C|)$ exceeding the classical fidelity $|C|/N$ 
	iff $|C|/N<1$ and there exists a CPTP map $\mc{E}\colon B\to C$ such that $\omega_{AC}=\id_A\otimes\mc{E}(\rho_{AB})$ violates the reduction criterion.
\end{thm}

\subsection{Accessible information}

The accessible information is a measure of how much classical information can be retrieved from a quantum state ensemble $(\rho_x)_x$ with probabilities $(p_x)_x$ \cite{holevo1973bounds}.
To define it, denote by $H(X)=-\sum_i p_i \log p_i$ the Shannon entropy of a probability distribution $(p_i)_i$, and denote by $I(X;Y) = H(X)+H(Y)-H(XY)$ the mutual information of a pair of random variables $(X,Y)$, with the logarithm taken to base 2.
For a quantum state ensemble $\mathcal{E} = (p_x,\rho_x)_{x=1}^N$ and a POVM $M=\lbrace M_x\rbrace_{x=1}^N$, the accessible information $I(\cE,M)$ is defined as the mutual information $I(X;X')$, where $X$ takes values $x\in[N]$ with probability $p_{x}$, and $X'$ describes the outcome of sampling $\rho_{x}$ from $X$ and measuring with respect to $M$.
That is, the random variable pair $(X,X')$ has distribution $p_{XX'} = p_{X'|X}(x'|x) p_{X}(x)$ where $p_{X'|X}(x'|x) = \tr(M_{x'}\rho_{x})$ (compare this to \eqref{Eq:transition-prob}) and $p_{X}({x})=p_{x}$.
Optimizing $I(X;X')$ over the measurement $M$ yields the accessible information $I(\mathcal{E}) \coloneqq \max_{M} I(\mathcal{E},M)$ of the ensemble $\cE$.

We can now state the following quantitative relationship between teleportation and dense coding:

\begin{thm}\label{thm:densecoding}
	Let $(\rho_{AB},\lbrace\Pi^i\rbrace_{i=1}^N,\lbrace\cD^i\rbrace_{i=1}^N)$ be a teleportation protocol for which $\Pi$ is the optimal measurement maximizing the entanglement fidelity $F$ in \eqref{eq:fidelity} for the given set $\lbrace\cD^i\rbrace_i$ of decoding operations.
	Denote by $p\equiv \psucc = d^2F/N$ with $d\equiv |C|$ the success probability of the corresponding state discrimination problem of the ensemble $\cE = (1/N,\omega_{AC}^i)_{i=1}^N$.
	Then the accessible information $I(\cE)$ of the associated dense coding protocol satisfies
	\begin{multline}
		\log N - (1-p)\log(N-1) - h(p) \\ \leq I(\cE) \leq \log N + \log p,\label{eq:acc-Np}
	\end{multline}
	where $h(p) = -p\log p -(1-p)\log(1-p)$ is the binary entropy of $p$.
	In terms of the entanglement fidelity $F\equiv F(\Lambda)$ with $\Lambda$ the teleportation channel defined in \eqref{eq:teleportation-channel},
	\begin{multline}
		2\log d + \log F + (1-p)\left[\log(1-p)-\log(d^2F-p)\right] \\ \leq I(\cE) \leq 2\log d + \log F.\label{eq:acc-dF}
	\end{multline}
\end{thm}

\begin{proof}
	We first prove the upper bound in \eqref{eq:acc-Np}.
	The optimal success probability of distinguishing the states $\lbrace \omega_{AC}^i\rbrace_{i=1}^N$ can be expressed as \cite{koenig2009operational}
	\begin{align}
		p = \exp\left(-H_{\mathrm{min}}(X|AC)_\rho\right).\label{eq:p-Hmin}
	\end{align}
	Here, the min-entropy 
	\begin{align}
		H_{\mathrm{min}}(A|B)_\rho \coloneqq - \inf_{\sigma_B} \inf\lbrace \lambda\colon \rho_{AB} \leq 2^\lambda \one_A\otimes\sigma_B\rbrace,
	\end{align} 
	with the first infimum over all states $\sigma_B$, 
	is evaluated on the classical-quantum state
	\begin{align}
		\rho_{XAC} = \sum_{i=1}^N \frac{1}{N} |i\rangle\langle i|_X\otimes \omega_{AC}^i.\label{eq:cq}
	\end{align}
	Recall furthermore that the min-entropy satisfies the relation
	\begin{align}
		H_{\mathrm{min}}(A|B) \leq H(A|B) = H(AB) - H(B),\label{eq:min-vN}
	\end{align}
	where $H(A)_\rho = -\tr\rho_A\log\rho_A$ denotes the von Neumann entropy of a quantum state $\rho_A$.
	
	The data-processing inequality with respect to the measurement $AC\to Y$ gives
	\begin{align}
		I(\mathcal{E}) \equiv I(\mathcal{E},\Pi) = I(X;Y) \leq I(X;AC)_\rho.
	\end{align}
	Furthermore, using the fact that $H(X)_\rho = \log N$ for our ensemble $\mathcal{E}$ together with \eqref{eq:p-Hmin} and \eqref{eq:min-vN}, we obtain
	\begin{align}
		I(X;AC)_\rho &= H(X)_\rho - H(X|AC)_\rho\notag\\
		& = \log N - H(X|AC)_\rho\notag\\
		&\leq \log N - H_{\mathrm{min}}(X|AC)_\rho\notag\\
		&= \log N + \log p,
	\end{align}
	which proves the upper bound in \eqref{eq:acc-Np}.
	
	For the lower bound, we employ Fano's inequality \cite{fano1952inequality}: Let $X,Y$ be random variables with joint distribution $p_{XY}(x,y)$, and let $\tilde{X} = f(Y)$ be a random variable over the same alphabet as $X$. 
	Defining the decoding error probability $p_e = \operatorname{Prob}(\tilde{X}\neq X)$, we have
	\begin{align}
		H(X|Y) \leq h(p_e) + p_e \log(|X|-1).
	\end{align}
	Choosing $X$ and $Y$ as in the paragraph above Thm.~\ref{thm:densecoding} and using $p = 1 - p_e$ as well as $h(p)=h(1-p)$, we obtain
	\begin{align}
		I(\cE) = I(X;Y) &= H(X) - H(X|Y)\notag\\ &\geq \log N - (1-p)\log(N-1) - h(p).
	\end{align}
	This concludes the proof of eq.~\eqref{eq:acc-Np}.
	The bounds in \eqref{eq:acc-dF} follow using the relation $F = Np/d^2$ from Lem.~\ref{lem:fidelity}.
\end{proof}

Since $\psucc = \Fcl$ by Thm.~\ref{Prop:duality}, the bounds in Thm.~\ref{thm:densecoding} relate our two chosen figures of merit, the classical correlation fidelity $\Fcl$ and the accessible information $I(\cE)$.

Thm.~\ref{thm:densecoding} generalizes some of the results of Werner~\cite{werner2001teleportation} on the duality between teleportation and dense coding protocols to arbitrary entanglement assistance.
To see this, let us first assume that we are in the ``tight'' regime $N=d^2$, where $N$ is the number of measurement outcomes and $d\equiv |C|$ is the dimension of the system to be teleported.
Then \eqref{eq:fidelity} implies that $F = \frac{N}{d^2}p = p$.
Hence, the teleportation protocol has perfect fidelity $F=1$ if and only if the associated state discrimination problem has perfect success probability $p=1$.
In this case, $I(\cE) = 2\log d$ by Thm.~\ref{thm:densecoding}, the maximal value of any dense coding protocol using a $|C|$-dimensional quantum channel.

In the more general case when $N$ and $d^2$ may be different, the second inequality in \eqref{eq:acc-dF} implies that any dense coding protocol achieving the maximal accessible information $I(\cE) = 2\log d$ corresponds to a perfect teleportation protocol with $F=1$.
Interestingly, the converse statement is not necessarily true, which can for example be observed in port-based teleportation protocols achieving a fidelity of $F$ arbitrarily close to 1 when $d$ is fixed and $N$ is large \cite{ishizaka2008asymptotic,ishizaka2009quantum,beigi2011simplified,studzinski2017port,mozrzymas2018optimal,christandl2021asymptotic}.
Here, the accessible information of the associated dense coding protocol vanishes in the limit $N\to\infty$ \cite{ishizaka2015remarks,strelchuk2021minimal}.

For certain port-based teleportation protocols, the following bound on $I(\cE)$ was stated (without proof) in \cite{strelchuk2021minimal}:
\begin{align}
	I(\cE) \leq \frac{d^2}{N}F \log d^2 + \log F.\label{eq:strelchuk-ub}
\end{align}
Using \eqref{eq:fidelity} we have $\frac{d^2}{N}F = p \leq  1$ (as a probability), and hence \eqref{eq:strelchuk-ub} may improve upon the upper bound on $I(\cE)$ in \eqref{eq:acc-dF}.
On the other hand, the bounds derived in Thm.~\ref{thm:densecoding} hold for arbitrary teleportation and dense coding protocols.

Thm.~\ref{thm:densecoding} and Cor.~\ref{cor:classical-fidelity} also provide an alternative and self-contained proof of the result of \cite{horodecki2001noisy} that bound entangled states do not provide an advantage in dense coding: 
From Cor.~\ref{cor:classical-fidelity}, any teleportation protocol using a bound entangled state has fidelity at most $1/d$.
It then follows from the upper bound \eqref{eq:acc-dF} of Thm.~\ref{thm:densecoding} that $I(\cE) \leq 2 \log d + \log F \leq 2\log d - \log d = \log d$.
Hence, we have proved the following statement:
\begin{cor}[{\cite{horodecki2001noisy}}]
	Any dense coding protocol defined in terms of a bound entangled state has accessible information at most $\log d$, where $d$ is the dimension of the system sent through the noiseless quantum channel.
\end{cor}

\section{Conclusion}
In this work provide a new perspective on the dual tasks of teleportation and dense coding.
We argue that the known operational duality between these tasks extends to the level of protocols: the measurement and decoding operations in a teleportation protocol from Alice to Bob can be repurposed as a dense coding protocol from Bob to Alice.
The quantitative link between these tasks is provided by generalizing an expression for the teleportation fidelity in port-based teleportation to fully general protocols. 
This fidelity expression reformulates teleportation as a state discrimination problem that Bob and Alice aim to solve in the associated dense coding protocol.
We use this connection to give two different quantitative duality theorems between teleportation and dense coding: one in terms of a classical version of fidelity called ``classical correlation fidelity'', and one in terms of the information-theoretically relevant accessible information.
We also give new proofs of the established facts that bound entangled states do not provide any advantage over classical resources in either teleportation or dense coding.
In fact, we strengthen these results by showing that for both protocols a bipartite resource state can give a quantum advantage if and only if there exists a locally processed version that violates the reduction criterion.


In the dense coding scenario we assumed that the quantum channel connecting Bob to Alice is noiseless. 
It would be interesting to consider how the correlation fidelity changes in the presence of channel noise.  
A similar problem was recently studied in \cite{Chitambar-2023a}, which considered a setting where the noisy channel is fixed and Alice and Bob are allowed to choose an optimal entangled state $\ket{\varphi}_{AB}$ for the classical communication task.  
The optimal classical correlation fidelity $F_{\text{cl}}$ in this case can then be identified as a function of the given channel, known as its entanglement-assisted \textit{communication value}.  
We can thus interpret the classical correlation fidelity in \eqref{eq:opt-Fcl} as the entanglement-assisted communication value of a noiseless channel, with the assistance restricted to a non-optimal entangled resource state $\rho_{AB}$.

Another interesting future direction of research is to study teleportation in the more general setting of two-way LOCC, where classical communication is allowed both from Alice to Bob and vice versa.

\paragraph*{Acknowledgments}
We thank Jamie Sikora for helpful feedback. This research was supported by a grant through the IBM-Illinois Discovery Accelerator Institute.

\printbibliography[heading=bibintoc]

\end{document}